\documentclass{article}

\usepackage{amssymb,amsmath}
\usepackage{amsthm}
\usepackage{bbm} 
\usepackage{amsfonts,mathrsfs}
\usepackage{url}
\usepackage{enumerate}
\usepackage{txfonts}

\def\ot{\otimes}
\def\D{\textsf{D}}\def\H{\textsf{H}}\def\S{\textsf{S}}\def\T{\textsf{T}}
\newcommand{\inner}[2]{\langle #1 , #2\rangle}
\newcommand{\innerm}[3]{\langle #1 | #2 | #3\rangle}
\newcommand{\out}[2]{| #1\rangle\langle #2 |}

\newcommand{\trans}{{\scriptstyle\mathsf{T}}}

\newcommand{\pa}[1]{(#1)}

\newcommand{\set}[1]{\{#1\}}
\newcommand{\Set}[1]{\left\{#1\right\}}

\newcommand{\ket}[1]{|#1\rangle}

\def\Jamiolkowski{J}
\newcommand{\jam}[1]{\Jamiolkowski\pa{#1}}

\DeclareMathOperator{\vectorize}{vec}
\newcommand{\col}[1]{\vectorize\pa{#1}}
\newcommand{\row}[1]{\vectorize\pa{#1}^{\dagger}}

\DeclareMathOperator{\trace}{Tr}
\newcommand{\ptr}[2]{\trace_{#1}\pa{#2}}

\newcommand{\tr}[1]{\ptr{}{#1}}

\newcommand{\fontmapset}{\mathbf} 

\newcommand{\mset}[2]{\fontmapset{#1}\pa{#2}}

\newcommand{\lin}[1]{\mset{L}{#1}}

\newcommand{\identity}{\mathbbm{1}}
\newcommand{\idsup}[1]{\identity_{#1}}

\newcommand{\Natural}{\mathbb{N}}

\newcommand{\Complex}{\mathbb{C}}

\def\cH{\mathcal{H}}
\def\cK{\mathcal{K}}

\newtheorem{thrm}{Theorem}[section]

\newtheorem{prop}[thrm]{Proposition}

\theoremstyle{definition}
\newtheorem{definition}[thrm]{Definition}

\numberwithin{equation}{section}

\begin{document}

\title{The dynamical additivity and the strong dynamical additivity \\of quantum operations}
\author{Zhang Lin and Wu Junde\\ \small
Department of Mathematics, Zhejiang University, 310027 Hangzhou,
P.~R.~China} \date{E-mail: linyz@zju.edu.cn,
wjd@zju.edu.cn}\maketitle

\begin{abstract}
In the paper, the dynamical additivity of bi-stochastic quantum
operations is characterized and the strong dynamical additivity is
obtained under some restrictions.\\~\\
PACS numbers: 02.10.Ud, 03.67.-a, 03.65.Yz

\end{abstract}


\section{Introduction}

In quantum information theory, there are two well-known entropic
inequalities for quantum-mechanical systems, that is, subadditivity
of entropy of bipartite quantum state $\rho^{AB}$:
$$
\S(\rho^{AB}) \leqslant \S(\rho^A) + \S(\rho^B)
$$
and strong subadditivity of entropy of tripartite quantum state
$ABC$:
$$
\S(\rho^{ABC}) + \S(\rho^B) \leqslant \S(\rho^{AB}) + \S(\rho^{BC}),
$$
where $\rho^X(X = A,B,AB,BC)$ are the reduction to corresponding
system $X$. In quantum information processing, one are especially
interested in the extreme cases of quantum states, for instance,
under what conditions the subadditivity or strong subadditivity
inequality of entropy of quantum states are saturated? By the
Pinsker's inequality \cite{Pinsker}:
$$
\S(\rho^A) + \S(\rho^B) - \S(\rho^{AB}) \geqslant \frac{1}{2\ln
2}\left(\left\|\rho^{AB} - \rho^A \ot \rho^B\right\|_1\right)^2,
$$
where $\|*\|_1$ are the trace-norm, it follows that $\S(\rho^{AB}) =
\S(\rho^A) + \S(\rho^B)$ if and only if $\rho^{AB} = \rho^A \ot
\rho^B$. This resolves the saturation of subadditivity inequality.
Compared with the subadditivity, the more complicated construction
that follows will give the solution to the saturation of the strong
subadditivity inequality. The description is as follows
\cite{Hayden}: a tripartite state $\rho^{ABC}$ are such that
$\S(\rho^{AB}) + \S(\rho^{BC}) = \S(\rho^{ABC}) + \S(\rho^B)$ if and
only if there is a decomposition of Hilbert space $\cH^B$ which is
used to describe the system $B$:
$$
\cH^B = \bigoplus_k \cH^L_k \ot \cH^R_k
$$
into a direct sum of tensor products such that
$$
\rho^{ABC} = \bigoplus_k  p_k \rho^{AL}_k \ot \rho^{RC}_k,
$$
where $\rho^{AL}_k$ is a state on $\cH^A \ot \cH^L_k$ and
$\rho^{RC}_k$ is a state on $\cH^R_k \ot \cH^C$ for each index $k$
and $\set{p_k}$ is a probability distribution.

Similarly these problems above-mentioned can be considered in the
regime of quantum operations. Let $\Phi$, $\Lambda$ and $\Psi$ be
three stochastic quantum operations (the notations will be explained
later) on a quantum system space $\cH$. The study on the behavior of
map entropy of composition of stochastic quantum operations is an
important and interesting problem. Recently Roga \emph{et. al.}
\cite{Roga} obtained that if $\Phi$ is bi-stochastic, then we have
the \emph{dynamical subadditivity}:
$$
\S(\Phi \circ \Psi) \leqslant \S(\Phi) + \S(\Psi).
$$
Moreover, if $\Phi$, $\Lambda$ and $\Psi$ are all bi-stochastic,
then we have the \emph{strong dynamical subadditivity}:
$$
\S(\Phi \circ \Lambda \circ \Psi) + \S(\Lambda) \leqslant \S(\Phi
\circ \Lambda) + \S(\Lambda \circ \Psi).
$$

In this paper, motivated by the structure of states which saturate
the inequality of strong subadditivity of quantum entropy, we
discuss under what conditions the dynamical subadditivity and the
strong dynamical subadditivity can be saturated, that is, the
dynamical additivity and the strong dynamical additivity. Firstly,
by using entropy-preserving extensions of quantum states, a
characterization of dynamical additivity of bi-stochastic quantum
operations is obtained. Next, we show that if quantum operations are
local operations \cite{Chiribella,D'Ariano} and have some kind of
orthogonality, then the strong dynamical additivity is also true.


\section{Preliminaries}

In this section we clarify the notations used in our paper.
Throughout the paper, only finite-dimensional Hilbert spaces $\cH$
are considered. Let $\lin{\cH}$ be the set of all linear operators
from $\cH$ to $\cH$. A state $\rho$ of some quantum system,
described by $\cH$, is a positive semi-definite operator of trace
one, in particular, for each unit vector $\ket{\psi} \in \cH$, the
operator $\rho = \out{\psi}{\psi}$ is said to be a \emph{pure
state}. The set of all states on $\cH$ is denoted by $\D(\cH)$.

If $X, Y \in \lin{\cH}$, then $\inner{X}{Y} = \tr{X^{\dagger}Y}$
defines an inner product on $\lin{\cH}$, which is called the
\emph{Hilbert-Schmidt inner product}. It is easily seen that if $X,
Y\in\lin{\cH}$ are two positive semi-definite operators, $X$ and $Y$
are orthogonal, i.e., $\inner{X}{Y} = 0$, if and only if $XY=0$.

Let $S,T\in\lin{\cH_1 \ot \cH_2}$ be two positive semi-definite
operators, where $\cH_1 = \cH_2 = \cH$. Denote $Y_1=\ptr{2}{Y},
Y_2=\ptr{1}{Y}(Y = S,T)$. Then $S_{1}, T_{1}, S_{2},
T_{2}\in\lin{\cH}$ are all positive semi-definite operators. If
$S_1T_1=S_2T_2=0$, then $S$ and $T$ are said to be
\emph{bi-orthogonal} \cite{Herbut}. Thus the notion of a state
decomposition that is bi-orthogonal is defined as follows:

\begin{definition}(\cite{Herbut})\label{Herbut}
Let $\rho^{AB}$ be a bipartite state in $\D(\cH_A \ot \cH_B)$. The
following state decomposition $ \rho^{AB} = \sum_k p_k \rho^{AB}_k$,
where $\rho^{AB}_k \in \D(\cH_A \ot \cH_B)$ for each $k$ and
$\set{p_k}$ is a probability distribution with each $p_k>0$, is
called \emph{bi-orthogonal} if, in terms of the reductions of
$\rho^{AB}_k$, $ \rho^X_k \rho^X_{k'} = 0 (X = A, B; \forall k \neq
k')$.
\end{definition}


Let $\set{\ket{m}}$ be the standard basis for $\cH_2$,
correspondingly $\set{\ket{\mu}}$ for $\cH_1$. For each $P =
\sum_{m,\mu}p_{m\mu} \out{m}{\mu} \in \lin{\cH_1,\cH_2}$, if we
denote $\col{P} = \sum_{m,\mu}p_{m\mu}\ket{m\mu}$, then
$\mathbf{vec}$ defines a simple correspondence between
$\lin{\cH_1,\cH_2}$ and $\cH_2 \ot \cH_1$. Moreover, if $\cH_A$ and
$\cH_B$ are two Hilbert spaces, $\set{\ket{m}}$ and
$\set{\ket{\mu}}$ are their standard bases, respectively, then we
can also define a map $\mathbf{vec}$ over a bipartite space that
describes a change of bases from the standard basis of
$\lin{\cH_{A}\ot\cH_{B}}$ to the standard basis of $\cH_{A} \ot
\cH_{A} \ot \cH_{B} \ot \cH_{B}$, that is,
$$
\col{\out{m}{n} \ot \out{\mu}{\nu}} = \ket{mn}\ot\ket{\mu\nu}.
$$
The following properties of the $\mathbf{vec}$ map are easily
verified \cite{Watrous}:
\begin{enumerate}
\item The $\mathbf{vec}$ map is a linear bijection. It is also an
isometry, in the sense that
$$
\inner{X}{Y} = \inner{\col{X}}{\col{Y}}
$$
for all $X,Y \in
\lin{\cH_1,\cH_2}$.

\item For every choice of operators $A \in \lin{\cH_1,\cK_1}, B \in
\lin{\cH_2,\cK_2}$, and $X \in \lin{\cH_2,\cH_1}$, it holds that
$$
(A \ot B) \col{X} = \col{A X B^\trans}.
$$

\item For every choice of operators $A, B \in
\lin{\cH_1, \cH_2}$, the following equations hold:
\begin{eqnarray*}
\ptr{1}{\col{A}\row{B}} & = & AB^\dagger, \\ \ptr{2}{\col{A}\row{B}}
 & = & (B^\dagger A)^\trans.
\end{eqnarray*}

\item If $X\in\lin{\cH_{A}}$, $Z\in\lin{\cH_{B}}$, then $\col{X\ot
Z}=\col{X}\ot\col{Z}$ .
\end{enumerate}


Denote by $\T(\cH)$ the set of all \emph{linear super-operators}
from $\lin{\cH}$ to $\lin{\cH}$. For each $\Phi\in\T(\cH)$, it
follows from the Hilbert-Schmidt inner product of $\lin{\cH}$ that
there is a linear super-operator $\Phi^{\dagger}\in\T(\cH)$ such
that $\inner{\Phi(X)}{Y}=\inner{X}{\Phi^{\dagger}(Y)}$ for any $X,
Y\in \lin{\cH}$. $\Phi^{\dagger}$ is referred to the \emph{dual
super-operator} of $\Phi$.


$\Phi\in\T(\cH)$ is said to be \emph{completely positive} (CP) if
for each $k\in\Natural$,
$\Phi\ot\idsup{M_{k}(\Complex)}:\lin{\cH}\ot
M_{k}\pa{\Complex}\to\lin{\cH}\ot M_{k}\pa{\Complex}$ is positive,
where $M_{k}\pa{\Complex}$ is the set of all $k\times k$ complex
matrices. It follows from the famous theorems of Choi \cite{Choi}
and Kraus \cite{Kraus1} that $\Phi$ can be represented in the
following form: $\Phi=\sum_{j}Ad_{M_{j}}$, where
$\{M_j\}_{j=1}^n\subseteq \lin{\cH}$, that is,
$\Phi(X)=\sum_{j=1}^nM_jXM_j^\dagger,\,\, X\in \lin{\cH}$.
Throughout this paper, $\dagger$ means the adjoint operation of an
operator. Moreover, if $\{M_j\}_{j=1}^n$ is pairwise orthogonal,
then $\Phi=\sum_{j}Ad_{M_{j}}$ is said to be a canonical
representation of $\Phi$. In \cite{Choi,Kraus2}, it was proved that
each quantum operation has a canonical representation.


The so-called \emph{quantum operation} on $\cH$ is just a CP and
trace non-increasing super-operator $\Phi\in \T(\cH)$, moreover, if
$\Phi$ is CP and trace-preserving, then it is called
\emph{stochastic}; if $\Phi$ is stochastic and unit-preserving, then
it is called \emph{bi-stochastic}.


The famous \emph{Jamio{\l}kowski isomorphism}
$J:\T(\cH)\longrightarrow\lin{\cH\ot\cH}$ transforms each
$\Phi\in\T(\cH)$ into an operator $\jam{\Phi}\in\lin{\cH\ot\cH}$,
where
$\jam{\Phi}=\Phi\ot\idsup{\lin{\cH}}(\col{\idsup{\cH}}\row{\idsup{\cH}})$.
If $\Phi\in \T(\cH)$ is CP, then $\jam{\Phi}$ is a positive
semi-definite operator, in particular, if $\Phi$ is stochastic, then
$\frac{1}{N}\jam{\Phi}$ is a state on $\cH\ot\cH$, we denote the
state by $\rho(\Phi)$, \cite{Bengtsson}.


The information encoded in a quantum state $\rho\in\D\pa{\cH}$ is
quantified by its \emph{von Neumann entropy}
$\S\pa{\rho}=-\tr{\rho\log_{2}\rho}$. If $\Phi\in \T(\cH)$ is a
stochastic quantum operation, we denote the von Neumann entropy
$\S(\rho(\Phi))$ of $\rho(\Phi)$ by $\S(\Phi)$ which is called
\emph{map entropy},  $\S(\Phi)$ describes the decoherence induced by
the quantum operation $\Phi$.


\section{Entropy-Preserving Extensions of Quantum States and the Dynamical Additivity}

The technique of quantum state extension without changing entropy is
a very important and useful tool. It is employed by Datta to
construct an example which shows equivalence of the positivity of
quantum discord and strong subadditivity for quantum mechanical
systems. Based on this fact, Datta obtained that zero discord states
are precisely those states which satisfy the strong additivity for
quantum mechanical systems \cite{Datta}. In what follows, we will
use it to give a characterization of dynamical additivity of map
entropy.

The next proposition is concerned with one type of quantum state
extensions without changing entropy.

\begin{prop}
Let $\rho \in \D(\cH)$. If $\set{\ket{i}}$ is a basis for $\cH$ and
$\rho = \sum_{i,j=1}^{N}\rho_{i,j}\out{i}{j}$, then
$\widetilde{\rho} = \sum_{i,j=1}^{N}\rho_{i,j}\out{ii}{jj}$ is a
state in $\D(\cH \ot \cH)$, and $\S(\widetilde{\rho}) = \S(\rho)$.
\end{prop}

\begin{proof}
By the spectral decomposition theorem, $\rho =
\sum_{k}\lambda_{k}\out{x_{k}}{x_{k}}$, where $\lambda_{k} \geqslant
0$, $\set{\ket{x_{k}}}$ is an orthonormal set for $\cH$. This
implies that $\rho_{i,j} = \innerm{i}{\rho}{j} =
\sum_{k}\lambda_{k}\langle i\out{x_{k}}{x_{k}}j\rangle =
\sum_{k}\lambda_{k}x^{(i)}_{k}\bar{x}^{(j)}_{k}$. Note that
$\set{\ket{x_{k}}}$ is an orthonormal set for $\cH$, so
$\sum_{i=1}^{N}x^{(i)}_{m}\bar{x}^{(i)}_{n} = \delta_{mn}$. Now
\begin{eqnarray*}
\widetilde{\rho} & = &
\sum_{i,j=1}^{N}(\sum_{k}\lambda_{k}x^{(i)}_{k}\bar{x}^{(j)}_{k})\out{i}{j}\ot\out{i}{j}
 = \sum_{k}\lambda_{k}(\sum_{i,j=1}^{N}x^{(i)}_{k}\bar{x}^{(j)}_{k}\out{i}{j}\ot\out{i}{j})\\
& = &
\sum_{k}\lambda_{k}(\sum_{i=1}^{N}x^{(i)}_{k}\ket{ii})(\sum_{i=1}^{N}x^{(i)}_{k}\ket{ii})^\dagger
= \sum_{k}\lambda_{k}\col{X_{k}}\row{ X_{k}},
\end{eqnarray*}
where $\col{X_{k}}=\sum_{i=1}^{N}x^{(i)}_{k}\ket{ii}\in \cH\ot\cH$.
Moreover, it is easy to show that $\row{ X_{m}}\col{X_{n}} =
\delta_{mn}$, thus $\widetilde{\rho}$ is a state on $\cH \ot \cH$.
It is obvious that $\S(\widetilde{\rho})=\S(\rho)$.
\end{proof}

Let $\Lambda\in\T(\cH)$ be stochastic. If $\Lambda$ has two Kraus
representations
$\Lambda=\sum_{p=1}^{d_{1}}Ad_{S_{p}}=\sum_{q=1}^{d_{2}}Ad_{T_{q}}$,
$\rho\in\D(\cH)$, take two Hilbert spaces $\cH_1$ and $\cH_2$ such
that dim $\cH_1=d_1$, dim $\cH_2=d_2$, $\set{\ket{m}}$ and
$\set{\ket{\mu}}$ are the base of $\cH_{1}$ and $\cH_{2}$,
respectively. Define
$$
\gamma_{1}(\Lambda)  =  \sum_{m,n=1}^{d_1}\tr{S_{m}\rho
S_{n}^\dagger}\out{m}{n}\mbox{\ and\ } \gamma_{2}(\Lambda)  =
\sum_{\mu,\nu=1}^{d_{2}}\tr{T_{\mu}\rho
T_{\nu}^\dagger}\out{\mu}{\nu}.
$$
Then $\gamma_{k} \in \D(\cH_k)(k = 1,2)$, and
$\S(\gamma_{1}(\Lambda))=\S(\gamma_{2}(\Lambda))$.

In fact, without loss of generality, we may assume $d_{1}=d_{2}=d$.
Then there exists a $d \times d$ unitary matrix $U=[u_{m\mu}]$ such
that for each $1\leq m\leq d$,
$S_{m}=\sum_{\mu=1}^{d}u_{m\mu}T_{\mu}$. Thus
\begin{eqnarray*}
\sum_{m,n=1}^{d} \tr{S_{m}\rho S_{n}^\dagger} \out{m}{n} & = &
\sum_{m,n=1}^{d} \tr{(\sum_{\mu=1}^{d} u_{m\mu}T_{\mu}) \rho
(\sum_{\mu=1}^{d}u_{n\mu}T_{\mu})^\dagger} \out{m}{n}\\
& = & U\left[\sum_{\mu,\nu=1}^{d} \tr{T_{\mu} \rho
T_{\nu}^\dagger}\out{m}{n}\right]U^\dagger.
\end{eqnarray*}
Let $V:\cH_1\longrightarrow\cH_2$ be a unitary operator such that
$V\ket{m} = \ket{\mu}$. Then
$$
\sum_{m,n=1}^{d}\tr{S_{m}\rho
S_{n}^\dagger}\out{m}{n} =
UV\left[\sum_{\mu,\nu=1}^{d}\tr{T_{\mu}\rho
T_{\nu}^\dagger}\out{\mu}{\nu}\right]V^\dagger U^\dagger,
$$
which implies that $\gamma_{1}$ and $\gamma_{2}$ are unitarily
equivalent and thus the conclusion follows \cite{Petz}.

For each stochastic $\Lambda\in\T(\cH)$ and $\rho\in\D(\cH)$, we
denote $\S(\rho;\Lambda)$ by $\S(\gamma_{1}(\Lambda))$. It follows
from the above discussion that $\S(\rho;\Lambda)$ is well-defined
\cite{Lindblad}. Moreover, it is easy to see that if
$\rho=\frac{1}{N}\identity$, then $\S(\rho;\Lambda)=\S(\Lambda)$,
\cite{Roga}.

It follows from above that if $\Phi, \Psi\in \T(\cH)$ are two
bi-stochastic quantum operations, $\Phi=\sum_{m=1}^{N^2}Ad_{S_{m}}$
and $\Psi=\sum_{\mu=1}^{N^2}Ad_{T_{\mu}}$ are their canonical
representations, respectively. Taking a $N^2$ dimensional complex
Hilbert space $\cH_{0}$, for each $\rho\in\D(\cH)$, we define
$$\gamma(\Phi\circ\Psi) = \sum_{m,n,\mu,\nu=1}^{N^2}\tr{S_{m}T_{\mu}\rho(S_{n}T_{\nu})^\dagger}
\out{m\mu}{n\nu},$$ then $\gamma(\Phi\circ\Psi)$ is a state on
$\cH_0\ot\cH_0$, and when $\rho=\frac{1}{N}\identity$,
$\S(\gamma(\Phi\circ\Psi))=\S(\Phi\circ\Psi)$, that is, $\S(\rho,
\Phi\circ\Psi) = \S(\Phi\circ\Psi)$.

Our mail result of this section is the following:

\begin{thrm}\label{Th:additivity}
Let $\Phi, \Psi\in \T(\cH)$ be two bi-stochastic quantum operations,
$\Phi(\rho)=\sum_{m=1}^{N^2}Ad_{S_{m}}$ and
$\Psi=\sum_{\mu=1}^{N^2}Ad_{T_{\mu}}$ be their canonical
representations, respectively. Then
$\S(\Phi\circ\Psi)=\S(\Phi)+\S(\Psi)$ if and only if
$\tr{S_{m}T_{\mu}(S_{n}T_{\nu})^\dagger}=\frac{1}{N}\tr{S_{m}
S_{n}^\dagger}\tr{T_{\mu}T_{\nu}^\dagger}$; i.e.,
$\inner{S_{n}T_{\nu}}{S_{m}T_{\mu}}=\frac{1}{N}\inner{S_{n}}{S_{m}}\inner{T_{\nu}}{T_{\mu}}$
for all $m,n,\mu,\nu=1,\ldots,N^2$.
\end{thrm}

\begin{proof} The Jamio{\l}kowski isomorphisms of $\Phi$ and $\Psi$ are
$$
\jam{\Phi} = \sum_{m=1}^{N^2} \col{S_{m}}\row{S_{m}}, \jam{\Psi} =
\sum_{\mu=1}^{N^2} \col{T_{\mu}}\row{T_{\mu}},
$$
respectively, where $\inner{\col{S_{m}}}{\col{S_{n}}} =
s_{m}\delta_{mn}$ and $\inner{\col{T_{\mu}}}{\col{T_{\nu}}} =
t_{\mu}\delta_{\mu\nu}$. For each $\rho \in \D(\cH)$, let
$$\gamma(\Phi\circ\Psi) = \sum_{m,n,\mu,\nu=1}^{N^2} \tr{S_{m}T_{\mu}\rho(S_{n}T_{\nu})^\dagger}\out{m\mu}{
n\nu} = \sum_{m,n,\mu,\nu=1}^{N^2} \tr{S_{m}T_{\mu}\rho(S_{n}T_{\nu})^\dagger}\out{m}{n}\ot\out{\mu}{\nu}.$$
Then we have
\begin{eqnarray*}
\gamma(\Psi) & = & \sum_{\mu,\nu=1}^{N^2}\tr{T_{\mu}\rho
T_{\nu}^\dagger}\out{\mu}{\nu}=\ptr{1}{\gamma(\Phi\circ\Psi)},\\
\gamma(\Phi) & = & \sum_{m,n=1}^{N^2}\tr{S_{m}\rho
S_{n}^\dagger)}\out{m}{n} = \ptr{2}{\gamma(\Phi\circ\Psi)}.
\end{eqnarray*}
Note that when $\rho=\frac{1}{N}\identity$,
$\S(\gamma(\Phi\circ\Psi))=\S(\Phi\circ\Psi)$,
$\S(\gamma(\Psi))=\S(\Psi)$ and $\S(\gamma(\Phi))=\S(\Phi)$. Thus,
we have
\begin{eqnarray*}
\S(\Phi\circ\Psi)=\S(\Phi)+\S(\Psi)&\Leftrightarrow&\S(\gamma(\Phi))+\S(\gamma(\Psi))=\S(\gamma(\Phi\circ\Psi))\\\
&\Leftrightarrow&\gamma(\Phi\circ\Psi)=\gamma(\Phi)\ot\gamma(\Psi)\\
&\Leftrightarrow&\tr{S_{m}T_{\mu}(S_{n}T_{\nu})^\dagger}=\frac{1}{N}\tr{S_{m}
S_{n}^\dagger}\tr{T_{\mu}T_{\nu}^\dagger}\\&&=\frac{s_{m}t_{\mu}}{N}\delta_{mn}\delta_{\mu\nu}(\forall
m,n,\mu,\nu=1,\ldots,N^2).
\end{eqnarray*}
\end{proof}


\section{Bi-orthogonal Decomposition and Strong Dynamical Additivity}

In order to study the strong dynamical additivity, we need the
following bi-orthogonality and the bi-orthogonal decomposition of
quantum operations.

Let $\Phi, \Psi\in\T\pa{\cH}$ be CP super-operators. If their
Jamio{\l}kowski isomorphisms $\jam{\Phi}$ and $\jam{\Psi}$ are
bi-orthogonal, then $\Phi$ and $\Psi$ are said to be
\emph{bi-orthogonal}.

\begin{prop}\label{bi-ortho}
If $\Phi=\sum_{\mu}Ad_{M_{\mu}}$, $\Psi=\sum_{\nu}Ad_{N_{\nu}}$,
then $\Phi$ and $\Psi$ are bi-orthogonal if and only if
$M^\dagger_{\mu}N_{\nu}=0$ and $M_{\mu}N^\dagger_{\nu}=0$ for all
$\mu$ and $\nu$, if and only if $\Phi\circ\Psi^{\dagger}=0$ and
$\Phi^{\dagger}\circ\Psi=0$, if and only if
$\Psi\circ\Phi^{\dagger}=0$ and $\Psi^{\dagger}\circ\Phi=0$.
\end{prop}

\begin{proof}
Note that $\jam{\Phi}=\sum_{\mu}\col{M_{\mu}}\row{M_{\mu}}$,
$\jam{\Psi}=\sum_{\nu}\col{N_{\nu}}\row{N_{\nu}}$. By the
definition, $\Phi$ and $\Psi$ are bi-orthogonal if and only if
$\jam{\Phi}$ and $\jam{\Psi}$ are bi-orthogonal, i.e.,
$\ptr{2}{\jam{\Phi}}\ptr{2}{\jam{\Psi}} =
\ptr{1}{\jam{\Phi}}\ptr{1}{\jam{\Psi}} =0$. Since
\begin{eqnarray*}
\ptr{2}{\jam{\Phi}}\ptr{2}{\jam{\Psi}} & = &
\Set{\sum_{\mu}M_{\mu}{M_{\mu}}^\dagger}\Set{\sum_{\nu}N_{\nu}{N_{\nu}}^\dagger}
= \sum_{\mu,\nu}M_{\mu}{M_{\mu}}^\dagger N_{\nu}{N_{\nu}}^\dagger
\end{eqnarray*}
and
\begin{eqnarray*}
\ptr{1}{\jam{\Phi}}\ptr{1}{\jam{\Psi}} & = &
\Set{\sum_{\mu}[{M_{\mu}}^\dagger M_{\mu}]^{\trans}}
\Set{\sum_{\nu}[{N_{\nu}}^\dagger N_{\nu}]^{\trans}} =
\sum_{\mu,\nu}[{M_{\mu}}^\dagger M_{\mu}]^{\trans}
[{N_{\nu}}^\dagger N_{\nu}]^{\trans},
\end{eqnarray*}
it follows that both $\jam{\Phi}$ and $\jam{\Psi}$ are bi-orthogonal
if and only if $M_{\mu}{M_{\mu}}^\dagger N_{\nu}{N_{\nu}}^\dagger=0$
and ${M_{\mu}}^\dagger M_{\mu}{N_{\nu}}^\dagger N_{\nu}=0$ for all
$\mu$ and $\nu$, if and only if $M^\dagger_{\mu}N_{\nu}=0$ and
$M_{\mu}N^\dagger_{\nu}=0$ for all $\mu$ and $\nu$.
\end{proof}

By mimicking the Definition \ref{Herbut}, we introduce the following
notion of bi-orthogonal decomposition for CP super-operator:
\begin{definition}
A CP super-operator $\Phi \in \T(\cH)$ has a \emph{bi-orthogonal
decomposition} if $\jam{\Phi}$ has a bi-orthogonal decomposition:
$\jam{\Phi} = \sum_{k}D_k$, where $\set{D_k}$ is a family of
pairwise bi-orthogonal positive semi-definite operators.
\end{definition}

If $\jam{\Phi}$ can be represented as a sum $\sum_{k}D_k$ of
pairwise bi-orthogonal positive semi-definite operators, decompose
each $D_k$ by the spectral decomposition theorem as
$$D_k = \sum_{i}d_k^{(i)}\col{\widetilde{M}^{(i)}_k}\row{\widetilde{M}^{(i)}_k} = \sum_{i}\col{M^{(i)}_k}\row{M^{(i)}_k},$$
where $M^{(i)}_k \in \lin{\cH}$, $\col{M^{(i)}_k} = \sqrt{d^{(i)}_k}
\col{\widetilde{M}^{(i)}_k}$ and $\inner{M^{(i)}_k}{M^{(j)}_k} =
d^{(i)}_k \delta_{ij}$, then $\Phi_k = \sum_{i} Ad_{M^{(i)}_k}$ as
$\jam{\Phi_k} = D_k$. Since $\trace_{2}{D_k} = \sum_{i}
M^{(i)}_{k}{M^{(i)}_{k}}^\dagger$ and $\trace_{1}{D_k} =
\sum_{i}[{M^{(i)}_k}^\dagger M^{(i)}_k]^\trans$, it follows form the
bi-orthogonality of $\set{D_k}$ that ${M^{(i)}_s}^\dagger M^{(j)}_t
= 0$ and $M^{(i)}_s{M^{(j)}_t}^\dagger = 0$ for any $s\neq t$ and
all sub-indices $i,j$. This implies that ${\Phi_m}^\dagger \circ
\Phi_n = 0$ and $\Phi_m \circ {\Phi_n}^\dagger = 0$ if $m\neq n$.

The following proposition can be viewed as a characterization of
$\Phi$ having a bi-orthogonal decomposition:
\begin{prop}
A CP super-operator $\Phi \in \T(\cH)$ has a bi-orthogonal
decomposition if and only if $\Phi = \sum_{k} \Phi_k$, where
$\set{\Phi_k}$ is a collection of CP super-operators in $\T(\cH)$
and ${\Phi_m}^{\dagger} \circ \Phi_n = 0$ and $\Phi_m \circ
{\Phi_n}^\dagger = 0$ for all $m\neq n$.
\end{prop}

By Proposition 1 in \cite{Skowronek}, it follows from the above
Proposition \ref{bi-ortho} that
\begin{enumerate}
\item For $i=1,2$, let $\Phi_i, \Psi_i \in \T(\cH)$ be CP super-operators, $\Phi_1$ and $\Phi_2$ be bi-orthogonal, and $\Psi_1$ and $\Psi_2$ be
bi-orthogonal. Then for any CP super-operator $\Lambda\in\T(\cH)$,
$\Phi_1\circ\Lambda\circ\Psi_1$ and $\Phi_2\circ\Lambda\circ\Psi_2$
are also bi-orthogonal.

\item If $\Phi, \Psi \in \T(\cH)$ are CP and bi-orthogonal, then for any positive semi-definite
operators $X,Y \in \lin{\cH}$, $\Phi(X)$ and $\Psi(Y)$ are orthogonal.
\end{enumerate}

Our mail result of this section is the following:

\begin{thrm}\label{Th:strongadditivity}
Assume that $\Phi, \Lambda, \Psi \in \T(\cH)$ are CP and
bi-stochastic, and the following conditions hold:
\begin{enumerate}[(i)]

\item $\cH = \bigoplus_{k=1}^{K} \cH^L_k \ot \cH^R_k $,
where $\dim\cH^L_k = d^L_k, \dim\cH^R_k = d^R_k$ and $\sum_{k=1}^{K}
d^L_k d^R_k = N$;

\item $\Phi = \bigoplus_{k=1}^{K} \Phi^L_k \ot Ad_{U^R_k},
\Lambda = \bigoplus_{k=1}^{K} \Lambda^L_k \ot \Lambda^R_k$, and
$\Psi = \bigoplus_{k=1}^{K} Ad_{V^L_k} \ot \Psi^R_k$,\\
that is, $\Phi|_{\lin{\cH^L_k \ot \cH^R_k}} = \Phi^L_k \ot
Ad_{U^R_k}, \Psi|_{\lin{\cH^L_k \ot \cH^R_k}} = Ad_{V^L_k} \ot
\Psi^R_k$, and $\Lambda|_{\lin{\cH^L_k \ot \cH^R_k}} = \Lambda^L_k
\ot \Lambda^R_k$, $\Phi^L_k ,\Lambda^L_k \in \T(\cH^L_k)$ are
CP and bi-stochastic, $V^L_k \in \lin{\cH^L_k}$ are unitary operators,
$U^R_k \in \lin{\cH^R_k}$ are unitary operators and
$\Psi^R_k,\Lambda^R_k \in \T(\cH^R_k)$ are CP and bi-stochastic.
\end{enumerate}
Then we have the following strong dynamical additivity:
$$
\S(\Phi \circ \Lambda) + \S( \Lambda \circ \Psi)
 = \S(\Lambda) + \S(\Phi \circ \Lambda \circ \Psi).
$$
\end{thrm}

\begin{proof}
Since
$$\Phi\circ\Lambda\circ\Psi = \sum_{k=1}^{K} \Phi^L_k \circ \Lambda^L_k \circ
Ad_{V^L_k} \ot
 Ad_{U^R_k} \circ \Lambda^R_k \circ \Psi^R_k$$
is a bi-orthogonal decomposition of $\Phi \circ \Lambda \circ \Psi$,
it follows that
$$\rho(\Phi \circ \Lambda \circ \Psi) = \sum_{k=1}^{K} \lambda_{k} \rho(\Phi^L_k \circ \Lambda^L_k \circ
Ad_{V^L_k}) \ot \rho(Ad_{U^R_k} \circ \Lambda^R_k \circ \Psi^R_k),$$
where $\lambda_{k}=\frac1N d^L_k d^R_k$ for each $k$ and
$\sum_{k=1}^{K}\lambda_{k}=1$. Thus,
\begin{eqnarray*}
\S(\Phi\circ\Lambda\circ\Psi) & = & \H(\lambda) + \sum_{k=1}^{K} \lambda_{k} \S(\Phi^L_k \circ \Lambda^L_k \circ
Ad_{V^L_k}) + \sum_{k=1}^{K} \lambda_{k} \S(Ad_{U^R_k} \circ \Lambda^R_k \circ \Psi^R_k)\\
 & = & \H(\lambda) + \sum_{k=1}^{K} \lambda_{k} \S(\Phi^L_k \circ \Lambda^L_k) + \sum_{k=1}^{K} \lambda_{k} \S(\Lambda^R_k \circ \Psi^R_k).
\end{eqnarray*}
Similarly,
\begin{eqnarray*}
\S(\Phi\circ\Lambda) & = & \H(\lambda) + \sum_{k=1}^{K} \lambda_{k} \S(\Phi^L_k \circ \Lambda^L_k) + \sum_{k=1}^{K} \lambda_{k} \S(\Lambda^R_k),\\
\S(\Lambda\circ\Psi) & = & \H(\lambda) + \sum_{k=1}^{K} \lambda_{k} \S(\Lambda^L_k) + \sum_{k=1}^{K} \lambda_{k} \S(\Lambda^R_k \circ \Psi^R_k),\\
\S(\Lambda) & = & \H(\lambda) + \sum_{k=1}^{K} \lambda_{k} \S(\Lambda^L_k) + \sum_{k=1}^{K} \lambda_{k} \S(\Lambda^R_k),
\end{eqnarray*}
where $\H(\lambda) = -\sum_{k=1}^{K} \lambda_{k}\log_{2}\lambda_{k}$ is
the \emph{Shannon entropy} of
$\lambda = (\lambda_{1},\lambda_{2},\ldots,\lambda_{K})$. It follows
from these equalities that
$\S(\Phi \circ \Lambda) + \S(\Lambda\circ\Psi)
 = \S(\Lambda) + \S(\Phi \circ \Lambda \circ \Psi)$.
\end{proof}


\section{Concluding Remarks}

In a closed quantum system, clearly
$\Phi=Ad_{U},\Lambda=Ad_{V},\Psi=Ad_{W}$ can saturate the Strong
Dynamical Subadditivity(SDS), where $U,V,W$ are unitary operators.
Thus this is trivial case. More generally, in an open quantum
system, there is a complete different scenario. In order to saturate
the SDS, local operation---today commonly called \emph{no-signaling}
\cite{D'Ariano}---could be considered in this case. It can be seen
from the Theorem \ref{Th:strongadditivity} that when
$\Phi,\Lambda,\Psi$ are all local operations, then SDS is saturated
by no-signaling operations. Hence the underlying Hilbert space and
corresponding quantum operations can be viewed as  a bipartite space
and bipartite operations. Intuitively, a quantum operation is
no-signaling if it cannot be used by spatially separated parties to
violate relativistic causality, i.e., no-signaling quantum operation
jointly implemented by several parties that cannot use it to
communicate with each other. Therefore, the entropy of the composite
quantum operations is just changed locally. The sufficient condition
in Theorem \ref{Th:strongadditivity} is supported by the
impossibility of communicating by local operations. It is
conjectured that SDS cannot be saturated by non-local operations. So
looking for a necessary condition to Theorem
\ref{Th:strongadditivity} may be restricted within the set of
no-signalling operations.

A possible application can be expected by the following
consideration. Firstly, we recall some concepts for quantum states.
The so-called \emph{squashed entanglement} \cite{Christandl} are
proposed recently by Christandl \emph{et. al.} and some attractive
properties of it are established. Among all known entanglement
measures, squashed entanglement is the entanglement measure which
satisfies most properties that have been proposed as useful for an
entanglement measure \cite{Brandao}. The squashed entanglement is
related to the strong subadditivity of entropy for quantum states.
It is described by the quantity
$$
E_{sq}(\rho^{AB}) = \inf_E\set{\frac12 I(A;B|E) : \rho^{ABE} \mbox{\
extension of\ } \rho^{AB}},
$$
where
$$
I(A;B|E) = \S(\rho^{AE}) + \S(\rho^{BE}) - \S(\rho^{ABE}) -
\S(\rho^{E})
$$
is the quantum conditional mutual information of $\rho^{ABE}$, which
measures the correlations of two quantum systems relative to a third
one. One important property of $E_{sq}$ is that it is faithful
\cite{Brandao}: $E_{sq}(\rho^{AB}) = 0$ if and only if $\rho^{AB}$
is separable state. Based on this result, an approximate version of
the fact are obtained that states $\rho^{ABE}$ with zero conditional
mutual information $I(A;B|E)$ are such that $\rho^{AB}$ is
separable, that is, if a tripartite state has small conditional
mutual information, its $AB$ reduction is close to a separable
state. This problem is left open at the end of \cite{Hayden}. The
conditional mutual information $I(A;B|E)$ is also used to demarcates
the edges of quantum correlations by Datta \cite{Datta}.

The above developments motivate naturally us to consider analogous
problems for quantum operations. For instance, for the given quantum
operations $\Phi,\Lambda,\Psi \in \T(\cH)$, when they are all
bi-stochastic, the quantity
$$
I(\Phi;\Psi|\Lambda) = \S(\Phi \circ \Lambda) + \S(\Lambda \circ
\Psi) - \S(\Phi \circ \Lambda \circ \Psi) - \S(\Lambda)
$$
can be defined similarly, but unfortunately which is asymmetric with
respect to the pair $(\Phi,\Psi)$, compared with the quantum state
situation. This is clear since different composite ordering of
quantum operations lead to different magnitudes. Apparently, there
are two extreme quantity
$$
\sup_\Lambda \Set{I(\Phi;\Psi|\Lambda) : \Lambda \in \T(\cH) \mbox{\
being CP and bi-stochastic}},$$
and
$$
\inf_\Lambda \Set{I(\Phi;\Psi|\Lambda) : \Lambda \in \T(\cH) \mbox{\
being CP and bi-stochastic}}
$$
can be considered. They may signify the maximal/minimum capacity of
decoherence induced by the composition of quantum operations $\Phi$
and $\Psi$. We leave it for the future research.


\subsection*{Acknowledgement.} We thank anonymous referees for comments on an earlier version of this
paper.



\end{document}